\documentclass[conference]{IEEEtran}
\IEEEoverridecommandlockouts
\usepackage{cite}
\usepackage{amsmath,amssymb,amsfonts}
\usepackage{algorithm}
\usepackage{algorithmic}
\usepackage{graphicx}
\usepackage{textcomp}
\usepackage{xcolor}
\usepackage{amsthm}
\usepackage{multirow}
\usepackage{makecell}
\usepackage{subfigure}
\usepackage{bbm}
\usepackage{pifont}
\usepackage{tikz}
\def\BibTeX{{\rm B\kern-.05em{\sc i\kern-.025em b}\kern-.08em
    T\kern-.1667em\lower.7ex\hbox{E}\kern-.125emX}}

\usepackage{xcolor}

\usepackage{hyperref}
\hypersetup{hidelinks}

\let\sigproof\proof\let\proof\relax
\let\sigendproof\endproof\let\endproof\relax
\let\proof\sigproof
\let\endproof\sigendproof
\theoremstyle{sig}
\newtheorem{theorem}{Theorem}

\usepackage{etoolbox}
\makeatletter
\patchcmd{\@makecaption}
  {\scshape}
  {}
  {}
  {}
\makeatletter
\patchcmd{\@makecaption}
  {\\}
  {.\ }
  {}
  {}
\makeatother
\def\tablename{Table}

\begin{document}

\title{\textsc{Spin}: Accelerating Large Language Model Inference with Heterogeneous Speculative Models}

\author{
\IEEEauthorblockN{
Fahao Chen\IEEEauthorrefmark{1},
Peng Li\IEEEauthorrefmark{2},
Tom H. Luan\IEEEauthorrefmark{2},
Zhou Su\IEEEauthorrefmark{2}, and
Jing Deng\IEEEauthorrefmark{3}}
\IEEEauthorblockA{\IEEEauthorrefmark{1}School of Computer Science and Engineering, The University of Aizu, Japan}
\IEEEauthorblockA{\IEEEauthorrefmark{2}School of Cyber Science and
Engineering, Xi’an Jiaotong University, China}
\IEEEauthorblockA{\IEEEauthorrefmark{3}Department of Computer Science, University of North Carolina at Greensboro, USA}
% \IEEEauthorblockA{Corresponding Author: Peng Li\quad Email: pengli@xjtu.edu.cn}
\IEEEauthorblockA{Emails: c24fahao@u-aizu.ac.jp, pengli@xjtu.edu.cn, tom.luan@xjtu.edu.cn, zhousu@ieee.org, jing.deng@uncg.edu}
}

\maketitle

\begin{abstract}
Speculative decoding has been shown as an effective way to accelerate Large Language Model (LLM) inference by using a Small Speculative Model (SSM) to generate candidate tokens in a so-called speculation phase, which are subsequently verified by the LLM in a verification phase. However, current state-of-the-art speculative decoding approaches have three key limitations: handling requests with varying difficulty using homogeneous SSMs, lack of robust support for batch processing, and insufficient holistic optimization for both speculation and verification phases. In this paper, we introduce \textsc{Spin}, an efficient LLM inference serving system based on speculative decoding, designed to address these challenges through three main innovations. First, \textsc{Spin} improves token speculation by using multiple heterogeneous SSMs, with a learning-based algorithm for SSM selection that operates without prior knowledge of request difficulty. Second, \textsc{Spin} employs a request decomposition method to minimize batching overhead during LLM verification. Finally, \textsc{Spin} orchestrates speculation and verification phases by pipelining their executions on GPUs to achieve further acceleration. Experimental results demonstrate that \textsc{Spin} significantly outperforms state-of-the-art methods, achieving a performance increase of approximately 2.28$\times$.
\end{abstract}


\section{Introduction}
% Large Language Models (LLMs), such as the GPT series~\cite{brown2020language, achiam2023gpt}, have demonstrated remarkable performance across a variety of generative tasks, e.g., question answering~\cite{liu2021makes, zhuang2024toolqa} and task automation~\cite{zhang2019pretraining, wen2024autodroid}. These tasks relies on efficient LLM inference, which presents a significant challenge due to the substantial number of model parameters and the considerable computational resources required. In addition, LLM inference employs an autoregressive decoding approach, where the model iteratively generates new tokens based on input tokens and previously generated ones. This process inherently limits execution parallelism and leads to under-utilization of GPU resources, resulting in high inference latency.
Large Language Models (LLMs), such as the GPT series~\cite{brown2020language, achiam2023gpt}, have demonstrated remarkable performance across a variety of generative tasks~\cite{zhao2023survey, fu2022tcb, wen2024autodroid} and efficient LLM inference has obtain significant research attention~\cite{wang2023tabi, liu2023deja, Sarathi2024, DistServer2024, Llumnix2024, InfiniGen2024, chen2024giant}.
Recently, speculative decoding~\cite{chen2023accelerating, leviathan2023fast} has been proposed to accelerate LLM inference by using a small speculative model (SSM) in conjunction with the LLM, as shown in \autoref{fig:intro_motivation}(a). The speculative decoding process begins with the SSM generating candidate tokens in a \textit{speculation} phase. Subsequently, these candidate tokens are verified by the LLM in a \textit{verification} phase. The SSM operates with high speed due to its small size, and verification can be efficiently completed with a single forward pass of the LLM, allowing all speculative tokens to be verified in parallel.

Existing research efforts on speculative decoding can be classified into two primary approaches. The first approach aims to improve the speculation phase by launching multiple homogeneous SSM instances, as shown in \autoref{fig:intro_motivation}(b), to generate more tokens with a higher probability to be accepted by the LLM. For instance, SpecInfer~\cite{miao2024specinfer} and Medusa~\cite{cai2024medusa} use a tree-based speculation policy, generating multiple candidate tokens at the same position to increase the acceptance rate. EGALE~\cite{li2024eagle} enhances the SSM with additional inputs, specifically the intermediate results of some LLM layers, to further boost the acceptance rate. The second approach focuses on optimizing the verification phase by reducing the inference cost of the LLM. For example, LayerSkip~\cite{elhoushi2024layer} and EESD~\cite{liu2024speculative} employ a portion of the LLM as the SSM, allowing intermediate results from the speculation phase to be reused during verification, thus reducing the overall inference cost of the LLM.

%%%
\begin{figure}[t] 
\begin{center}
\includegraphics[width=\linewidth]{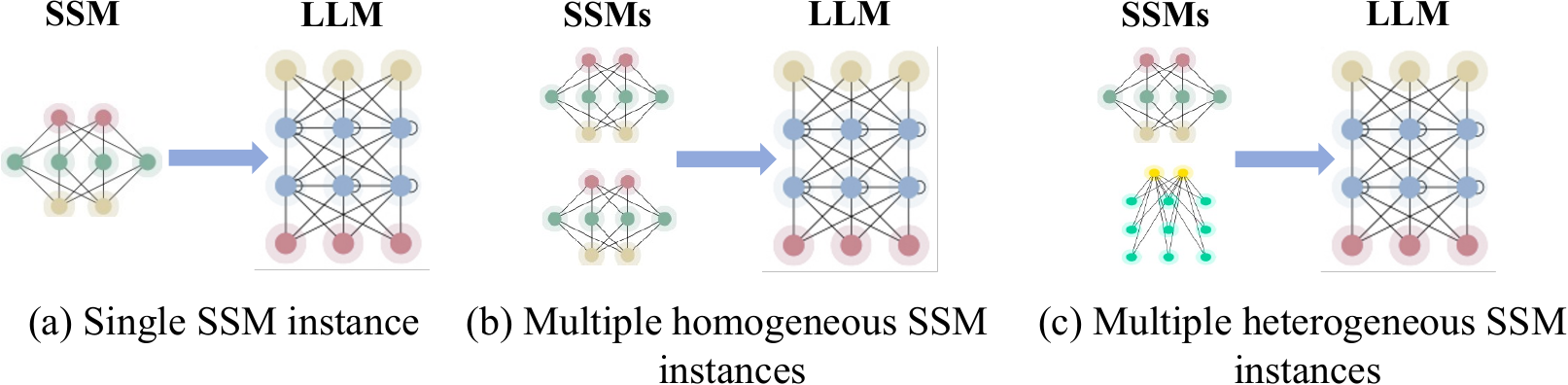}
\caption{\label{fig:intro_motivation}Illustration of different speculative decoding approaches.}
\end{center}
\end{figure}
%%%

Despite the potential of speculative decoding, existing works do not fully exploit its capabilities for accelerating LLM inference, primarily due to three reasons.
\underline{First}, during the speculation phase, existing works~\cite{chen2023accelerating, leviathan2023fast} typically employ homogeneous SSMs, 
%\footnote{In this paper, the term ``single SSM'' refers to a single type of SSM, which can have multiple instances deployed across different GPUs. This definition is used throughout the paper for simplicity.}, 
as shown in \autoref{fig:intro_motivation}(b), which is insufficient for handling requests of varying difficulty. Our observations indicate that some requests, such as commonsense questions, are relatively ``easy'' and can be effectively handled by smaller SSMs with a high probability of generating acceptable tokens. Conversely, some ``difficult'' requests require larger SSMs to produce candidate tokens that the LLM is more likely to accept.
\underline{Second}, during the verification phase, existing works do not robustly support batch-based verification by the LLM~\cite{chen2023accelerating, leviathan2023fast}. It has been shown that the LLM severely under-utilizes GPU resources when processing a single request. Batching multiple requests for parallel processing is a common practice to increase GPU utilization and throughput. However, our experimental results (\S\ref{motivation}) demonstrate that enabling batching requests during the verification phase degrades the speedup of speculative decoding by about 30\%. This degradation is primarily due to the varying lengths of requests, which necessitate padding with zeros to align them for GPU processing, thereby wasting GPU memory and computational resources.
\underline{Third}, existing works lack a holistic optimization approach for both speculation and verification phases. They either focus on improving the acceptance rate during the speculation phase or on accelerating the verification inference. This lack of comprehensive optimization prevents existing methods from efficiently orchestrating speculation and verification executions within a serving system, thus limiting the overall inference throughput.

\textbf{Goal:} In this paper, we present \textsc{Spin}, an efficient LLM inference serving system based on holistic-optimized speculative decoding by addressing the above limitations. \textsc{Spin} operates upon a GPU cluster and is capable of handling a variety of inference requests, such as multi-turn conversations and code generation. We propose to adopt multiple heterogeneous SSMs in the speculation phase, as shown in \autoref{fig:intro_motivation}(c), rather than relying on homogeneous SSMs, ensuring that each request is matched with the most appropriate SSM. Moreover, we propose a fast batch verification scheme for speculative decoding, which can significantly accelerate LLM verification while maintaining inference accuracy. 
%Finally, we orchestrate the speculation and verification phases by implementing a pipeline scheme, further enhancing GPU utilization and overall serving performance.

\textbf{Technical challenges:} 
Harnessing the potential of the aforementioned designs involves overcoming three critical challenges. The first one is how to select an appropriate SSM for each request, whose difficulty is unknown. The accuracy of speculative tokens depends on the current input and SSM capability, which is hard to be qualified or quantified before real execution. 
Second, simply batching requests for LLM verification introduces excessive padding tokens, a problem also encountered in traditional transformer models for language processing~\cite{fang2021turbotransformers, zhai2023bytetransformer}. A common solution is to batch requests with similar lengths, however, which is less effective here because speculative decoding exacerbates the length discrepancies between requests, and their lengths change dynamically.
Third, using heterogeneous SSMs may cause significant execution synchronization issues. In the default design, LLM verification cannot start until all candidate tokens are received from SSMs. Since these heterogeneous SSMs may work at different speeds, the slowest SSM can delay the entire system.

\textbf{Solution:} \textsc{Spin} addresses the above technical challenges with three novel designs:
\begin{itemize}
    \item We design a learning-based SSM selection algorithm, which enables \textsc{Spin} to choose suitable SSMs for inference requests to maximize speculative decoding performance (\S\ref{algorithm_idea}). We model the SSM selection problem as a multi-armed bandit and propose an efficient algorithm to explore each SSM's performance, without any knowledge about request difficulty or model features.
    \item \textsc{Spin} reduces batching overhead by decomposing long requests into short ones, which could be well aligned with fewer padded tokens. In order to guarantee the inference correctness after such decomposition, we propose a new attention computation scheme.
    \item To orchestrate speculation and verification phases for further accleration, \textsc{Spin} divides the original batch into micro-batches and pipelines their executions on SSMs and the LLM. This design significantly reduces the idle time of LLM during synchronization of heterogeneous SSMs.
\end{itemize}

We evaluate \textsc{Spin} using five SSMs from the LLaMA series, ranging from 68M to 1.4B. 
The extensive experiments on various LLMs and different datasets shows that \textsc{Spin} achieves 2.28$\times$ improvement on the inference throughput over existing baselines.

\section{Background and Motivation}
\subsection{Background}
\subsubsection{Autoregressive Decoding} 
The generative LLM inference typically uses an autoregressive decoding approach using the transformer architecture~\cite{vaswani2017attention}: the model iteratively processes the input sequence along with all previously generated tokens to produce the next token until it generates an ``end-of-sequence (EOS)'' token. 
% This entire generation involves two distinct phases: \textit{prefill} and \textit{decode}. During the prefill phase, the model processes the input prompt and generates the first output token. Subsequently, in the decode phase, output tokens are generated sequentially, where each token produced in the prior step is passed back through the model to generate the next token. 
The transformer model consists of multiple blocks, where each block includes a self-attention layer and a feed-forward network (FFN) layer. For a request with $n$ tokens, denoted by $S=[s_{1}, s_{2},..., s_{n}]$, each layer first applies projections on each token through three learnable linear matrices $W^{Q}, W^{K}, W^{V}$, to generate queries, keys, and values matrices:
\begin{align}
    Q=W^{Q}S, \mbox{ }K=W^{K}S, \mbox{ }V=W^{V}S.
\end{align}
Then, each layers calculate the output $O_{i}$ of the self-attention layer for each token by:
\begin{align}
    &O_{i} = \sum_{j=1}^{n}a_{i,j}V_{j}, \mbox{ } a_{i,j} = \frac{\mathcal{F}(Q_{i}, K_{j})}{\sum_{j=1}^{n}\mathcal{F}(Q_{i}, K_{j})}, \label{eq-Oa}
\end{align}
where $a_{i,j}$ is the attention score between token $s_{i}$ and $s_{j}$. Typically, we conduct dot-product attention with Softmax normalization by setting $\mathcal{F}(Q_{i}, K_{j}) = \exp(Q_{i}K_{j}^{T})$. The self-attention outputs are then projected by the FFN layer and forwarded to the next transformer block as inputs. The outputs of the final transformer block are a probability vector to mark out the most probable output tokens.

During the generation, the intermediate results (key and value) of generated tokens are used in the self-attention layer for generating the subsequent tokens. To avoid the redundant computation, existing inference engines~\cite{kwon2023efficient, yu2022orca, llumnix} store these intermediate results in the GPU memory and reuse them in the following computation, which is referred to as KV cache. 

% From a systems viewpoint, LLM inference is predominantly constrained by memory bandwidth, with the main latency bottleneck arising from the memory bandwidth of accelerators rather than their computational capacity. This bottleneck is intrinsic to the sequential nature of autoregressive decoding, where each forward pass necessitates transferring the entire model parameters from High-Bandwidth Memory (HBM) to the accelerator's cache, i.e., SRAM. This process yields only one token at a time, leading to underutilization of the computational capabilities of modern accelerators. 

\subsubsection{Speculative Decoding}
The entire speculative decoding process includes multiple iterations, where each iteration involves a \textit{speculation phase} and a \textit{verification phase}. The speculation phase is performed by the SSM to generate multiple candidate tokens. After that, the LLM verifies these candidate tokens in the verification phase. If these tokens are consistent with LLM's output, they are accepted to be a part of the final output. The speculation and verification phases will repeat until an EOS token is accepted by the LLM.
Since the speculation is performed by the SSM, which operates significantly faster than the LLM, the more tokens accepted by the LLM from the SSM, the greater the speedup for the LLM's inference process.

\begin{figure}[t] 
\begin{center}
\includegraphics[width=\linewidth]{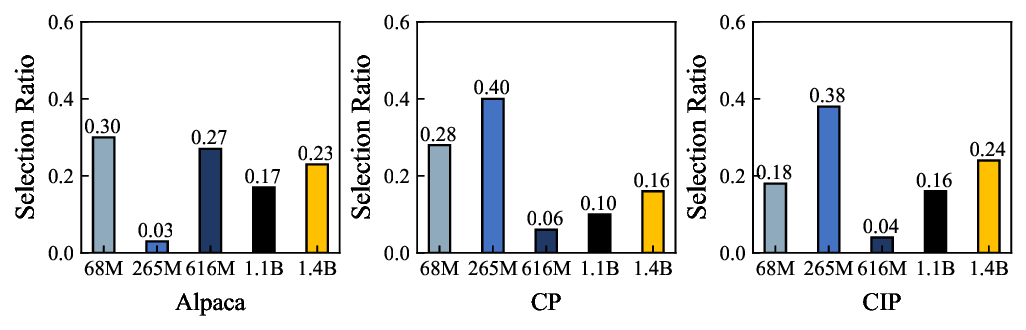}
\caption{\label{fig:single_SSM}The ratio of different SSMs selected as the best model across requests in three datasets.}
\end{center}
\end{figure}
%%%
%%%
\begin{figure}[t] 
\begin{center}
\includegraphics[width=\linewidth]{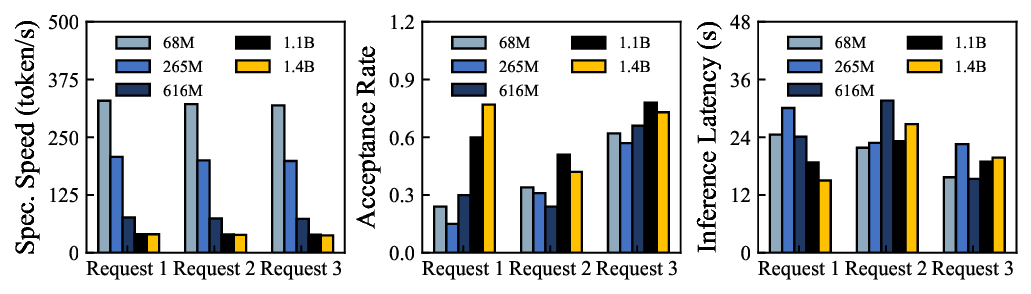}
\caption{\label{fig:challenge_selection}Results of three random requests using different SSMs.}
\end{center}
\end{figure}
\subsection{Motivation}\label{motivation}
\subsubsection{Limitation of Homogeneous SSMs}

Existing works mainly perform the speculative decoding using homogeneous SSMs, which fails to achieve the optimal benefits for all inference requests. We conduct experiments to demonstrate this issue. Specifically, we use LLaMA-7B as the LLM and select five SSMs from the LLaMA series, whose parameter scales range from 65 million to 1.4 billion. We evaluate the performance of speculative decoding across three datasets: Alpaca~\cite{alpaca}, ChatGPT Prompts (CP)~\cite{cp}, and Chatbot Instruction Prompts (CIP)~\cite{cip}. Inference latency is used as the performance metric.
\autoref{fig:single_SSM} reports the percentage of inference requests that achieve the best performance on each SSM. The results indicate that these SSMs play differently when handling these requests.
For example, only 3\% requests of the Alpaca dataset can achieve the best performance using the LLaMA-265M, but this model is preferred by 40\% requests of the CP dataset. 

To have deeper study about such differences, we randomly select 3 requests from the Alpaca dataset and show their speculation speed, acceptance rate, and inference latency under different SSMs in \autoref{fig:challenge_selection}. It can be seen that smaller SSM, e.g., LLaMA-68M, generates tokens much faster than larger ones, e.g., LLaMA-1.4B. Multiplying these with the acceptance rate, which is highly correlated to token prediction accuracy, the optimum SSM varies among these three requests.

% \textcolor{blue}{Then, we move to study the performance of speculative decoding processing within a single request. Specifically, we choose one requests from the CP dataset and split the speculative decoding process into five stages, where each stage generates 400 correct tokens. We find that even for a single request, the token generation difficulty can vary. For example, the token generation in the third stage seems to be hard and we need a larger SSM, while the fifth stage is easier and a smaller SSM is sufficient to achieve the optimal speculative decoding performance.}

The above results suggest that heterogeneous SSMs have great potential to accelerate speculative decoding. However, it is not easy to choose the best SSM for each request, because it is hard to accurately estimate token speculation capability of these SSMs. \textsc{Spin} addresses this challenge by proposing a learning-based algorithm that can perceive the request difficulty and choose appropriate SSMs accordingly. It does not rely on prior knowledge of SSM performance on requests.

%%%
\begin{figure}[t] 
\begin{center}
\includegraphics[width=\linewidth]{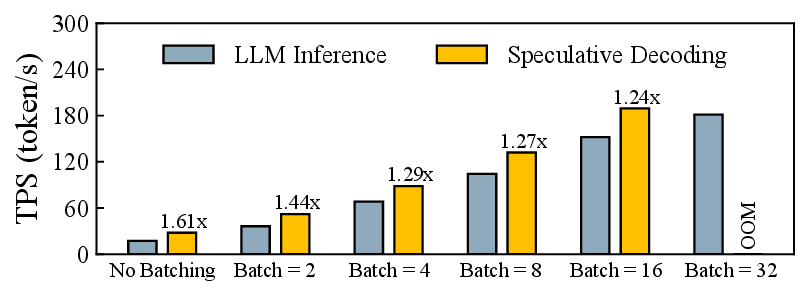}
\caption{\label{fig:moti_batch}The benefits of speculative decoding with different batch sizes.}
\end{center}
\end{figure}
%%%

%%%
\begin{figure}[t]
\begin{center}
\includegraphics[width=0.9\linewidth]{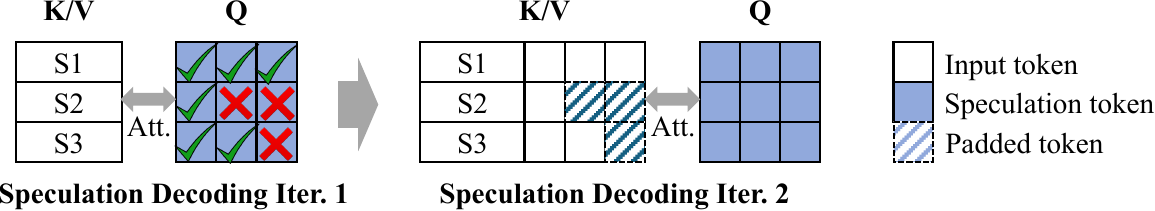}
\caption{\label{fig:moti_pad}Illustration of padding tokens in attention computation.}
\end{center}
\end{figure}
%%%
\subsubsection{Batching processing is not well supported by speculative decoding}

It has been shown that batching multiple requests can significantly improve GPU utilization and LLM inference throughput~\cite{yu2022orca}. However, when applying batching for speculative decoding, we find that it comes with increasing overhead as the batch size grows, which may overwhelm its benefits. As shown in \autoref{fig:moti_batch}, we measure the inference throughput, in terms of number of tokens output by LLM per second, under different batch sizes. Without batching, i.e., only 1 request is considered, speculative decoding can achieve 1.61$\times$ acceleration. However, such an acceleration decreases as we enlarge the batch size. Especially, when we set the batch size to 32, an out-of-memory error happens for speculative decoding while traditional inference can still work with growing throughput.

The reason of this inefficiency is that requests have different lengths, but CUDA kernels are designed to handle tensors with regular shapes. In order to align requests, additional padding tokens must be inserted to ensure tensors to be regular-shaped in the attention computation of the LLM forward pass. An example of padding tokens is shown in \autoref{fig:moti_pad}. These padding tokens contribute nothing to the final output, but waste GPU memory and computation resources.

Similar overhead of padding tokens also exists in traditional LLM inference without speculative decoding, as claimed by \cite{fang2021turbotransformers, zhai2023bytetransformer}. However, speculative decoding exacerbates length differences between batched requests since their token acceptance rates at the LLM may differ greatly, and thus more padding tokens are needed, leading to increased overhead. Moreover, although requests are with different length in traditional LLM inference, their length gaps are fixed once they are batched. In contrast, request lengths dynamically change in speculative decoding and thus that solutions proposed by \cite{fang2021turbotransformers, zhai2023bytetransformer} can hardly work here.

\begin{figure}[t] 
\begin{center}
\includegraphics[width=\linewidth]{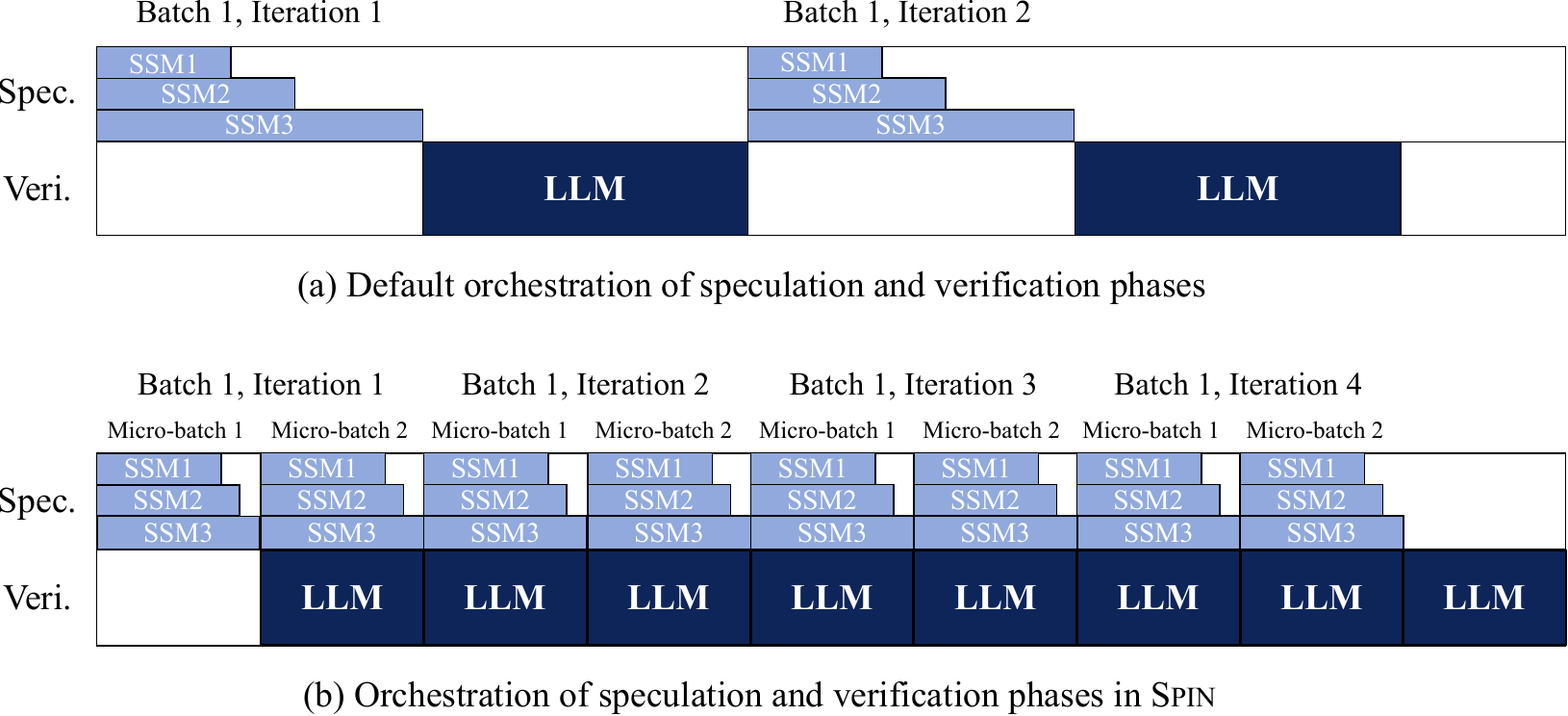}
\caption{\label{fig:moti_pipeline}The illustration of orchestration of speculation and verification.}
\end{center}
\end{figure}
%%%

\subsubsection{Speculation and verification need to be orchestrated} 
Token speculations of SSMs and token verification by the LLM are sequence-dependent, which could make GPUs idle and thus constrain the acceleration of speculative decoding, as shown in \autoref{fig:moti_pipeline}(a). Especially, heterogeneous SSMs have different running speeds and synchronizing their output would seriously postpone token verification for fast SSMs. SpecPIM \cite{li2024specpim} has proposed to let some speculative tokens skip the verification and directly use them in the next iteration. However, these tokens could be incorrect and waste the computing resources. In this paper, we propose to decompose the original batch into multiple micro-batches and pipeline their execution on SSMs and LLM. As shown in \autoref{fig:moti_pipeline}(b), we can significantly reduce the GPU idle time and guarantee the correctness of final output. However, determining how to effectively split the batch on each SSM into micro-batches is challenging, and we design an efficient heuristic to address this issue.

\section{System Overview}
%%%
\begin{figure}[t] 
\begin{center}
\includegraphics[width=\linewidth]{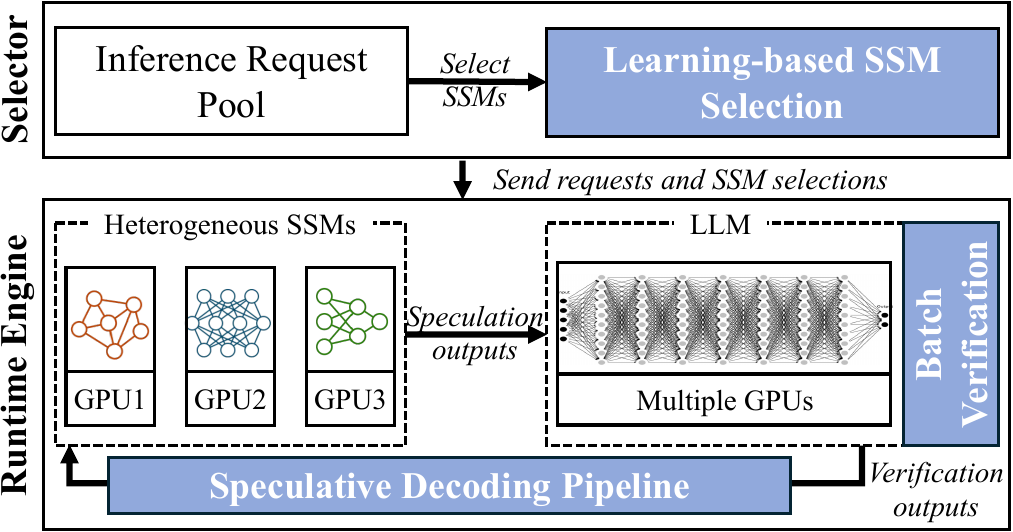}
\caption{\label{fig:overview}System Overview of \textsc{Spin}.}
\end{center}
\end{figure}
%%%

%\ding{182} \ding{183} \ding{184}
\textsc{Spin} is an LLM inference serving system based on speculative decoding~\cite{chen2023accelerating, leviathan2023fast} and batch processing~\cite{yu2022orca, kwon2023efficient}. The key design principle of \textsc{Spin} is to exploit heterogeneous SSMs to generate better speculative tokens for various requests.
An overview of \textsc{Spin} is shown in \autoref{fig:overview}, highlighting its two major modules. The first one is the \textsc{Spin}'s SSM selector (\S\ref{algorithm_idea}), which selects an appropriate SSM for each request using a learning-based selection algorithm that operates without prior knowledge of request difficulty. The second module is the runtime engine (\S\ref{engine}), responsible for managing the execution of heterogeneous SSMs and the LLM.
Requests are first assigned to SSMs, according to the decisions made by the selector. The speculative tokens generated by SSMs are sent to the LLM for verification. The request batching is enabled on both SSMs and LLM. To improve batching performance, \textsc{Spin} adopts a request decomposition method to reduce redundant padding tokens, particularly benefiting LLM verification. In addition, \textsc{Spin} uses a pipeline parallelism mechanism to orchestrate the speculation and verification phases for further acceleration. In practical deployment, we can launch multiple instances for each type of SSM, and each instance run on a dedicated GPU, because of the small size of SSM. The LLM operates on multiple GPUs using tensor parallelism, ensuring efficient utilization of GPU resources.

\section{Learning-based SSM Selection}\label{algorithm_idea}
% %%%
% \begin{figure}[t] 
% \begin{center}
% \includegraphics[width=\linewidth]{Figure/system model new.pdf}
% \caption{\label{fig:system_model}System model.}
% \end{center}
% \end{figure}
% %%%
% \subsection{Notations}
In this section, we first present the problem statement, followed by the algorithm design as well as performance analysis. 

\subsection{Problem Statement}
We consider a set of inference requests, denoted by $\mathcal{N}=\{1,2,...,N\}$, which is submitted to the GPU cluster managed by \textsc{Spin}.
There are multiple heterogeneous SSMs to perform speculation, denoted by the set $\mathcal{M}=\{1,2,...,M\}$. Each request $i \in \mathcal{N}$ selects only one SSM and multiple requests using the same SSM can be batched together for simultaneous processing. The batch size of each SSM $j$ is denoted as $B_j$. 
% \textcolor{blue}{In this work, we aim to maximize the benefits of the speculative decoding on each inference request.}
We define a metric of \textit{goodput}, which indicates how many tokens are accepted by LLM per second, to evaluate the performance of speculative decoding. It is important to note that the goodput does not equal the speculation speed of the SSM since the generated tokens may be rejected by the LLM. Formally, the goodput of the speculative decoding when using SSM $j$ for request $i$ is denoted as $g_{i,j}$. 

We denote $x_{i}\in \mathcal{M}$ as the SSM selection for request $i\in \mathcal{N}$. The objective of the SSM selection problem is to maximize the total goodputs for all requests:
\begin{align}
    \max_{x}\quad & \sum_{i=1}^{N}g_{i, x_{i}}\\
    \mbox{s.t.}\quad & |\mathcal{N}_{j}|\leq B_{j}, \forall j \in \mathcal{M};\label{cons_1}\\
    & x_{i}=\{1,2,...,M\}, \forall i \in \mathcal{N},
\end{align}
where $\mathcal{N}_{j}=\{i|x_{i}=j\}$ is the set of requests running on SSM $j$. Constraint (\ref{cons_1}) ensures the batch size limitation on each SSM. However, the above optimization problem cannot be directly solved because the goodput $g_{i, x_{i}}$ is strongly related to the acceptance rate at LLM, however, which is unknown before execution. 

%In addition, the algorithm complexity of SSM selection should be low enough so that it will not incur a high burden to make selection decisions. \textit{Therefore, we model the SSM selection as a multi-armed bandit (MAB) problem and develop an efficient algorithm based on exploration-exploitation to solve it.}

\subsection{Algorithm Design}\label{algorithm design}
%\noindent\textbf{Problem Transformation:} 
We model the SSM selection as a multi-armed bandit (MAB) problem. 
%The SSM selection for a request is analogous to pulling an arm in the bandit problem. 
Specifically, we divide the total speculative decoding process into a series of time slots, denoted by $\mathcal{T}=\{1,2,..., T\}$. At each time slot $t$, we decide the SSM selection for each request $i$, denoted by $x_{i}(t)\in \mathcal{M}$, and the corresponding real-time goodput is $r_{i,x(t)}$. Therefore, we have $g_{i,x_{i}} = \mathbb{E}[r_{i,x_{i}(t)}]$.
%Since the goodput $g_{i,j}$ is unknown, we can only observe a random throughput of accepted tokens at time slot $t$, which is denoted by $r_{i,j}(t)$ and with $g_{i,j} = \mathbb{E}[r_{i,j}(t)], \forall i\in \mathcal{N}, j\in\mathcal{M}$. We assume the throughput has lower and upper bounds, denoted by $r^{-}$ and $r^{+}$, respectively. 

The bandit-based algorithm for the SSM selection explores different SSMs to estimate the goodput, however, which incurs additional switching cost due to the re-computation of KV values of tokens already generated. The KV re-computation cost, which is denoted as $c_{i,j}(t)$, of request $i$ on SSM $j$ increases over time slots as more tokens are generated~\cite{CachedAttention}. We denote $z_{i}(t) = c_{i,x_{i}(t)}(t)\mathbbm{1}_{x_{i}(t)\neq x_{i}(t-1)}$, where $\mathbbm{1}_{A}$ is an indicator function for condition $A$. The cumulative regret to quantify the SSM selection performance can be expressed as:
\begin{align}
    &\mathcal{R}(T)\!=\!
    \underbrace{\sum_{i=1}^{N}(Tg_{i,x_{i}^{*}}- \sum_{t=1}^{T}\!\mathbb{E}[r_{i,x_{i}(t)}])}_{\text{Goodput regret}} + \lambda\underbrace{\sum_{t=1}^{T}\sum_{i=1}^{N}z_{i}(t)}_{\text{Switching cost}},
\end{align}
where $\lambda$ is a parameter to balance the goodput regret and switching cost. Then, the optimization problem can be re-formulated as:
\begin{align}
    \min_{x}\quad & \mathcal{R}(T) \label{regret_opt}\\
    \mbox{s.t.}\quad & |\mathcal{N}_{j}(t)|\leq B_{j}, \forall j \in \mathcal{M}, t\in \mathcal{T};\label{cons_2}\\
    & x_{i}(t)=\{1,2,...,M\}, \forall i \in \mathcal{N}, t\in \mathcal{T}.
\end{align}
where $\mathcal{N}_{j}(t)=\{i|x_{i}(t)=j\}$ is the set of requests running on SSM $j$ at time slot $t$.

\begin{figure}[t] 
\begin{center}
\includegraphics[width=\linewidth]{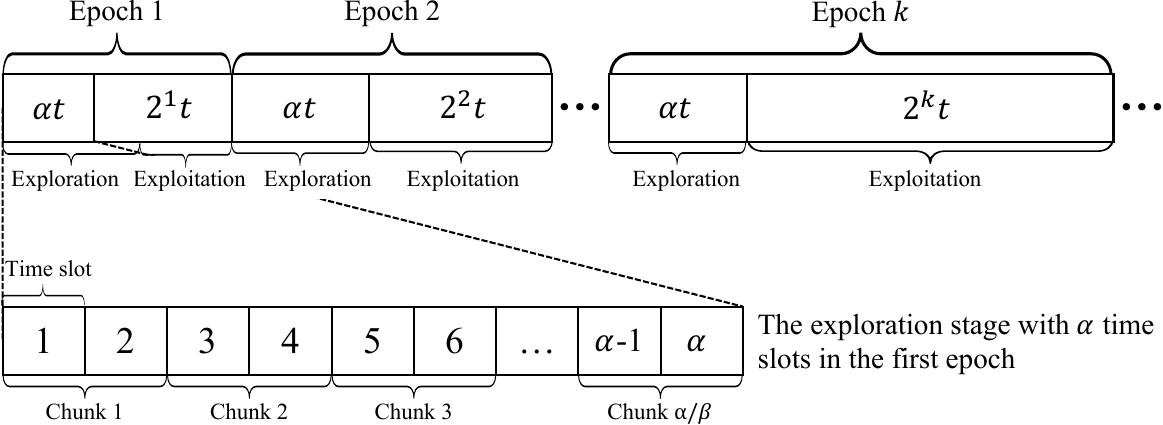}
\caption{\label{fig:epoch}The illustration of the epoch division. In each exploration stage, we further group time slots into chunks to reduce the SSM switching frequency. In this example, we set the chunk size $\beta$ as 2.}
\end{center}
\end{figure}

The proposed algorithm for solving the problem (\ref{regret_opt}) is presented in \autoref{alg}. We have no assumption about the length of $\mathcal{T}$, which is divided into epochs, as shown in \autoref{fig:epoch}. Each epoch includes two stages:

\noindent\textbf{Exploration Stage:} The exploration stage contains $\alpha$ time slots. In this stage, we randomly select SSMs for inference requests and observes the goodput $r_{i,j}(t)$. However, random selection would lead to frequent SSM switching with high cost. Therefore, we design a chunk-based exploration approach. Specifically, time slots in the exploration stage are further grouped into chunks. Each request $i$ uses the same SSM for all time slots within a time chunk, reducing the SSM switching frequency. We denote chunk size, i.e., the number of time slots in a chunk, as $\beta$. Thus, there are $\alpha/\beta$ time chunks in the exploration stage. At the end of each time slot $t$ in the exploration stage, request $i$ observes the throughput feedback $r_{i,j}(t)$. After the exploration stage at epoch $k$, we can get the average value of goodput for request $i$ on SSM $j$ as $\tilde{g}_{i,j}^{k}$.
% , where $G_{i,j}$ is the cumulative throughput feedback from SSM $j$ and $S_{i,j}$ is the total number of selections of SSM $j$.

%%%
\begin{algorithm}[t]
\caption{Learning-based SSM Selection (LBSS)}
\label{alg}
\begin{algorithmic}[1]
% \STATE \textbf{Initialization:} $R_{i,j}=0, S_{i,j}=0, \forall i\in \mathcal{N}, j\in \mathcal{M}$;
\FOR{epoch $k=1,2,...$}
    \STATE \textbf{Exploration Stage:} Estimate $\tilde{g}_{i,j}^{k}$ using \autoref{RS};
    \STATE \textbf{Exploitation Stage:} Select SSM by solving (\ref{matching_problem});
    \STATE Run requests on SSMs according to the decision in exploitation for $2^{k}$ time slots;
\ENDFOR
\end{algorithmic}
\label{tab1}
\end{algorithm}
%%
%%%
%%
\begin{algorithm}[t]
\caption{Chunk-based Exploration in Epoch $k$}
\label{RS}
\begin{algorithmic}[1]
\STATE \textbf{Initialization:} $t=0$;
\WHILE{$t\leq \alpha$}
    \FOR{$i=1,2,..., N$}
        \STATE Randomly select a SSM $j$;
        \STATE $x_{i}(t)=j$;
    \ENDFOR
    \FOR{$j=1,2,..., M$}
        \IF{$|\mathcal{N}_{j}|> B_{j}$}
            \STATE Randomly choose $B_{j}$ requests from $\mathcal{N}_{j}$ and drop others;
            % , as $\mathcal{N}_{j}^{-}$;
            % \STATE $x_{i}=0, \forall j\in \mathcal{N}_{j}^{-}$;
        \ENDIF
    \ENDFOR
    \STATE \emph{/*Chunk-based exploration*/}
    \FOR{$p=1,2,..., \alpha/\beta$}
        \STATE Run request $i$ on SSM $x_{i}(t)$ for $\beta$ time slots;
        \STATE Observe throughput at each time slot;
        % \STATE $R_{i,x_{i}(t)}+=r_{i,x_{i}(t)}(t)$ and $S_{i,x_{i}(t)}+=1$;
        % \STATE Update $t+=\beta$;
    \ENDFOR
\ENDWHILE
\STATE Estimate the goodput $\tilde{g}_{i,j}^{k}$ by averaging cumulative throughput;
\end{algorithmic}
\end{algorithm}

\noindent\textbf{Exploitation Stage:} We decide the best SSM for each inference request based on the estimated goodput $\tilde{g}_{i,j}^{k}$ by solving the following problem:
\begin{align}
    \max_{y} \quad &\sum_{i=1}^{N}\sum_{j=1}^{M}\tilde{g}_{i,j}^{k}y_{i,j},\label{matching_problem}\\
    \mbox{s.t.}\quad&\sum_{i=1}^{N}y_{i,j}\leq B_{j}, \forall j\in \mathcal{M};\\
    &y_{i,j}=\{0,1\}, \forall i\in \mathcal{N}, j\in\mathcal{M},
\end{align}
where $y_{i,j}$ is a binary variable indicating whether request $i$ is assigned to SSM $j$. The problem (\ref{matching_problem}) can be modeled as a bipartite graph matching problem between requests and SSMs. Since each SSM can accommodate a batch of requests, we can regard each SSM $j$ as having $B_{j}$ replicas, each capable of serving a request. The weight between the request and the SSM is the estimated goodput. The objective is to maximize the total weights. We use the KM algorithm to find a maximum-weight match $\pi^{k}$. It is important to note that $\pi^{k}$ is the best match based on the estimated goodput but may not indicate the optimal SSM selection since the estimated goodput could be inaccurate. According to the match $\pi^{k}$, each request is processed on the selected SSM for $2^{k}$ time slots. We use exponential duration for exploitation because as $k$ increases, the precision of goodput estimations improves, enhancing the probability of the optimal SSM selection.

% \begin{lemma} \label{lemma-1}
%     After $k$ epochs of exploration with $\alpha=\max\{\lceil \frac{eB^{2}}{2B_{min}^{2}}\rceil, \lceil \frac{3N^{2}(r^{+}\!-\!r^{-})^{2}}{2\Delta^{2}}\rceil\}$, the estimated goodput $\tilde{g}_{i,j}^{k}$ satisfies:
%     \begin{align}
%         \Pr(|\tilde{g}_{i,j}^{k} - g_{i,j}|> \frac{\Delta}{N})\leq 3NMe^{-k},
%     \end{align}
%     where $B=\sum_{j\in \mathcal{M}B_{j}}$ and $\Delta=\min_{x\neq x^{*}}\{\sum_{i=1}^{N}g_{i,x_{i}^{*}} - \sum_{i=1}^{N}g_{i,x_{i}}\}$ is the gap between the total goodput from the optimal SSM selection and the second optimal solution.
%     % After $k$ iterations of exploration and each exploration phase has $P=\max\{\lceil \frac{81B^{2}}{2B_{min}^{2}}\rceil, \lceil \frac{128N^{2}B(\rho^{+}\!-\!r^{L})^{2}}{9\Delta_{min}^{2}B_{min}}\rceil\}$ time slots, the exploration for mean reward estimation will be successful with the probability more than $1- 3NMe^{-k}$.
% \end{lemma}

% \begin{lemma}
%     The SSM selection in exploitation stage of epoch $k$ is sub-optimal with the probability of $P\leq 3NMe^{-k}$.
% \end{lemma}

\begin{theorem}
    The total regret $\mathcal{R}(T)$ is bounded by $\mathcal{O}(\log_{2}T)$.
\end{theorem}
\begin{proof}
    Please refer to the technical report~\cite{report} for the proof.
\end{proof}

\subsection{Fast SSM Switching}
We propose a fast SSM switching mechanism to further reduce the SSM switching cost during the exploration-exploitation process. The key insight is that newly speculative tokens cannot affect the KV of tokens already generated.
% \textcolor{red}{The key insight is that the inference iteratively concatenates the KV results of newly accepted tokens from the current speculation decoding step with the KV cache of previously generated tokens}. 
% In this way, the KV cache of accepted tokens remains unchanged with the addition of new tokens.
Therefore, we can pre-compute the KV cache of all existing tokens on the destination SSM, in parallel with the speculation on the source SSM.

% For example, as shown in \autoref{}, during the first speculation step on SSM1, we can also calculate the KV cache on SSM2 for the three input tokens since they remain constant. During the second step of speculation on SSM1, we can calculate the KV cache for the two accepted tokens on SSM2. When switching from SSM1 to SSM2, only the KV cache for the newly accepted token needs to be calculated, allowing speculative decoding to start earlier on SSM2.

To realize fast SSM switching, it is necessary to determine the destination SSM for switching in advance. During the exploration stage, the destination SSM can be easily determined because it is randomly selected. However, during the exploitation stage, the SSM selection is decided by a matching algorithm based on estimated goodput, making it challenging to determine the destination SSM in advance. A straightforward solution is to pre-compute the KV cache for all SSMs, but it introduces significant computation redundancy and wastes resources. To address this challenge, we choose the SSM with the highest estimated goodput for KV cache re-computation. Furthermore, \textsc{Spin} runs such KV cache re-computation during the idle time of SSMs, minimizing negative impact to the speculation on SSMs.

% \subsection{Algorithm Analysis}
% \begin{lemma} \label{lemma-1}
%     After $k$ epochs of exploration with $\alpha=\max\{\lceil \frac{eB^{2}}{2B_{min}^{2}}\rceil, \lceil \frac{3N^{2}(r^{+}\!-\!r^{-})^{2}}{2\Delta^{2}}\rceil\}$, the estimated goodput $\tilde{g}_{i,j}^{k}$ satisfies:
%     \begin{align}
%         \Pr(|\tilde{g}_{i,j}^{k} - g_{i,j}|> \frac{\Delta}{N})\leq 3NMe^{-k},
%     \end{align}
%     where $B=\sum_{j\in \mathcal{M}B_{j}}$ and $\Delta=\min_{x\neq x^{*}}\{\sum_{i=1}^{N}g_{i,x_{i}^{*}} - \sum_{i=1}^{N}g_{i,x_{i}}\}$ is the gap between the total goodput from the optimal SSM selection and the second optimal solution.
%     % After $k$ iterations of exploration and each exploration phase has $P=\max\{\lceil \frac{81B^{2}}{2B_{min}^{2}}\rceil, \lceil \frac{128N^{2}B(\rho^{+}\!-\!r^{L})^{2}}{9\Delta_{min}^{2}B_{min}}\rceil\}$ time slots, the exploration for mean reward estimation will be successful with the probability more than $1- 3NMe^{-k}$.
% \end{lemma}

% \begin{proof}
%     Please refer to the technical report~\cite{report} for the proof.
% \end{proof}

% \begin{lemma}
%     The SSM selection in exploitation stage of epoch $k$ is sub-optimal with the probability of $P\leq 3NMe^{-k}$.
% \end{lemma}
% \begin{proof}
%     Please refer to the technical report~\cite{report} for the proof.
% \end{proof}

% \begin{theorem}
%     The total regret $\mathcal{R}(T)$ is bounded by $\mathcal{O}(\log_{2}T)$ with the setting in Lemma \ref{lemma-1}.
% \end{theorem}
% \begin{proof}
%     Please refer to the technical report~\cite{report} for the proof.
% \end{proof}

\section{\textsc{Spin}'s Runtime Engine}\label{engine}

In this section, we present \textsc{Spin}'s speculative decoding engine, which includes a fast batch verification module, and a pipeline processing module. 

\subsection{Fast Batch Verification}
\begin{figure}[t] 
\begin{center}
\includegraphics[width=\linewidth]{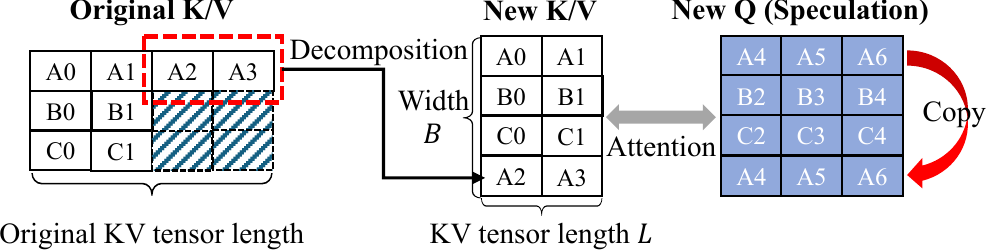}
\caption{\label{fig:decomposition}The illustration of the self-attention computation with request decomposition.}
\end{center}
\end{figure}
%JD5: what about the third box on the right (the blue-shaded one and the red arrow)? You never explain them in the discussion below.
%%%

The key idea of the fast batch verification is to reshape KV tensors by ripping off overlong parts of some requests and stitching them with short ones, so that we can reduce the padded tokens.
% decompose request inputs with long lengths into multiple shorter ones to align them with other requests. 
\autoref{fig:decomposition} illustrates a simple example of three requests, where tokens in white color are prompts, e.g., $(A0, A1, A2, A3)$, and the ones in blue color are candidate tokens, e.g., $(A4, A5, A6)$. In order to align KV tensors for self-attention computation, existing inference engines~\cite{fang2021turbotransformers, kwon2023efficient} insert
% to present the idea of fast batching verification. Consider a batch with three requests, where $A0$ is the prompt and $(A1, A2, A3)$ are four accepted tokens for request $A$. The numbers of accepted tokens vary from requests after a step of speculative decoding. For example, request $A$ accepts three tokens while request $B$ and request $C$ only accept one token each. The attention computation of the verification phase, which is a key component for Transformer-based LLM inference, requires aligned inputs for processing. Therefore, we need to insert 
four padding tokens into requests $B$ and $C$. In contrast, \textsc{Spin} decomposes request $A$ into two sub-requests so that it aligns with the other two requests without padding tokens.

However, decomposing requests could lead to incorrect computation results in the self-attention layer. Typically, the inference engine calculates the self-attention outputs for each request in the batch separately using (\ref{eq-Oa}). When a request is divided into multiple parts that locate in different rows, the denominator of (\ref{eq-Oa}) cannot summing all the $\mathcal{F}(Q_{i},K_{j})$ belonging to the same request, where $i$ and $j$ are indices of different tokens.
% When there are sub-requests belonging to the same request within the batch, the default inference engine can only calculate attention scores between tokens in each sub-request, and the attention scores between tokens from different sub-requests will be ignored, leading to incorrect inference results. 
% When there are different sub-requests in the batch that are decomposed from the same request, the above calculation could generate incorrect attention scores between tokens since the default self-attention calculation cannot consider them as belonging to the same request.
% For example, $D_{k}$ and $D_{k'}$ are sub-requests decomposed from the same request $\tilde{D}_{k}$. The correct attention score for tokens in $D_{k}$ or $D_{k'}$ should be $a_{i,j} = \frac{\mathcal{S}(Q_{i}, K_{j})}{\sum_{j\in \tilde{D}_{k}}\mathcal{S}(Q_{i}, K_{j})}, i,j\in D_{k} \mbox{ or }D_{k'}$. 

%JD5: but? it sounds like one sentence or one clause is misssing above. Somethig like "However, Eq.~\ref{eq-Oa} would generate different values for them without considering them as belonging to the same request."?

In order to ensure the correctness of self-attention computation, we modify the self-attention computation in Eq.~\eqref{eq-Oa} for request $S$ as follows:
\begin{align}
    &a_{i,j} = \frac{\mathcal{F}(Q_{i}, K_{j})}{\sum\limits_{j, r_{j}\in D}\mathcal{F}(Q_{i}, K_{j}) I_{j,S}},
\end{align}
where $S$ is a request before decomposition and $D$ is set of all tokens in the batch. We use the binary indicator $I_{j, S}$ to represent whether $r_{j}$ belongs to the request $S$.
Although the number of padding tokens in KV tensors can be reduced, it requires additional copies of tokens in the $Q$ tensor, as shown in \autoref{fig:decomposition}. 
%Thus, the request decomposing aims to achieve a trade-off, which minimizes the number of padding tokens, while decomposing the least requests. 
To balance the redundancy in $KV$ and $Q$ tensors, \textsc{Spin} carefully determines the request decomposition by setting the $KV$ tensor length $L$ and width $B$. The $Q$ tensor has the same width $B$ by replicating some rows accordingly. The length of $Q$ tensor is typically fixed by the prediction window, a system parameter of runtime engine, whose discussion is orthogonal to this work. We first fix tensor width $B$, which limits the overhead from $Q$ tensor as well as the number of decomposed requests. Then, we search for the $L$ that results in the minimum number of padding tokens in $KV$ tensor. Note that the value of $L$ does not have a linear relationship with the number of padding tokens, but this search process can be completed quickly.

%\textsc{Spin} selects the suitable batch length $L$ according to the insight: the smaller batch length leads to more decomposed requests. We first set the minimal batch length $L^{min}$, indicating the adopted batch length can not be smaller than $L^{min}$. The minimal batch length $L^{min}$ limits the number of decomposed requests, which can be adjusted according to the memory capacity. In addition, we denote $L^{max}$ as the maximal length of requests in the original batch. A batch length of more than $L^{max}$ will increase the total number of padded tokens. We find the batch length $L$ in $[L^{min}, L^{max}]$, which leads to the minimal number of padding tokens after request decomposition. 

% \noindent\textcolor{red}{
% \textbf{Discussion:} 
% % In current implementation of \textsc{Spin}, we only enable the request decomposition in the verification of the LLM since the complexity of SSMs is significantly smaller. Please note that the request decomposition can still benefit both the speculation and verification. 
% }

\subsection{Speculative Decoding Pipeline} \label{system_idea}

To reduce idle time caused by synchronizing heterogeneous SSMs, the speculative decoding pipeline module divides the original batch on each SSM into multiple micro-batches, so that speculation and verification can be performed on these micro-batches in a pipeline manner, as illustrated in \autoref{fig:moti_pipeline}(b). Specifically, after SSMs generate speculative tokens for a micro-batch, these results are sent to the LLM for verification. Meanwhile, SSMs can begin speculation on the next micro-batch, running in parallel with the LLM's verification. Since micro-batches involve lighter workloads for the SSMs, their speculation time is reduced, effectively minimizing the idle time of the LLM during synchronization of results from the SSMs.

A critical design choice of the speculative decoding pipeline is to decide micro-batch size to make the pipeline saturated to maximize resource utilization.
If the micro-batch size is too large, the idle time of the LLM due to synchronizing results from heterogeneous SSMs cannot be effectively reduced. On the other hand, while a small micro-batch size can effectively reduce idle time, the GPU resources are under-utilized during the LLM verification phase. 

\textsc{Spin} employs a simple heuristic to determine the optimal micro-batch setting on each SSM. Specifically, \textsc{Spin} iteratively splits the batch on each SSM into multiple micro-batches and checks whether the LLM can be fully utilized with the current setting. Specifically, in the initial iteration, each SSM obtains a batch of inference requests according to the SSM selection decision. Then, we split the batch on each SSM into $b_{0}$, e.g., $b_{0}=2$, micro-batches and check whether there is a obvious throughput drop. If there is no significant degradation, we continue to split the batch into $b_{0}+1$ micro-batches, which would further reduce synchronization time. This process repeats until a significant degradation in inference throughput is observed. It is important to note that we can offline profile the inference throughput of the LLM with different workloads so that we can quickly check it with different micro-batch settings, instead of running a verification execution and observing the throughput. Of course, we can also formulate the batch splitting as an optimization problem and get the optimal solution by solving it. However, we find such a heuristic algorithm works sufficiently well in practice and the algorithm overhead can be neglected. 

\section{Performance Evaluation}\label{evaluation}
\subsection{Experimental Settings}
\noindent\textbf{Model and Environment.} We evaluate the performance of \textsc{Spin} using the LLaMA family of models. We consider LLaMA-7B, LLaMA-13B, and LLaMA-30B as LLMs, respectively. For SSMs, we use five smaller models: LLaMA-68M, LLaMA-265M, LLaMA-616M, LLaMA-1.1B, and LLaMA-1.4B. All models are obtained from~\cite{huggingface}. We deploy each SSM on a single V100 GPU. For the LLM, we use a single V100 GPU for the LLaMA-7B model, 2xV100 GPUs for the LLaMA-13B model, and 4xV100 GPUs for the LLaMA-30B model. The multiple-GPU setting uses tensor parallelism. For software configuration, we use Ubuntu 20.04 with Linux Kernel version 5.15, CUDA 11.7, and cuDNN 8.6.0, along with the NVIDIA driver version 525.85. 

\noindent\textbf{Workloads.} We evaluate \textsc{Spin} using three public datasets: Alpaca~\cite{alpaca}, ChatGPT Prompts~\cite{cp}, and Chatbot Instruction Prompts~\cite{cip}. Aligning with the settings in  \cite{miao2024specinfer, wang2024minions}, we use the questions/prompts from these datasets to create workloads. Additionally, we combine the requests from all three datasets to form a mixed dataset to better simulate real-world traces, referred to as Mix.

\noindent\textbf{Baselines.} We compare \textsc{Spin} with the Vanilla speculative decoding inference system with homogeneous SSMs, which is commonly adopted by existing works~\cite{cai2024medusa, wang2024minions, li2024specpim}. We equip this Vanilla system with different SSMs in experiments.
% In addition, we enable batch verification in all Vanilla methods using the default solution, which pads tokens to align varying-length requests. 

%%
\begin{figure*}[t]
\centering
        \subfigure[Alpaca]{
            \begin{minipage}[b]{0.22\textwidth}
            \includegraphics[width=\textwidth]{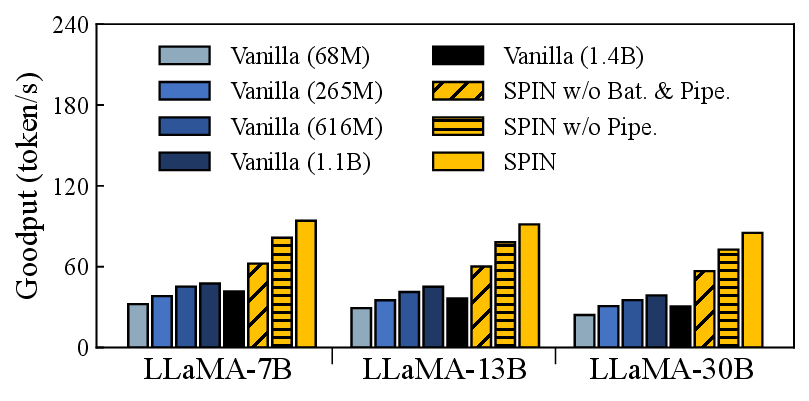}
            \end{minipage}
        \label{fig:ssm_alpaca}
            }
            \subfigure[CP]{
    		\begin{minipage}[b]{0.22\textwidth}
    		\includegraphics[width=\textwidth]{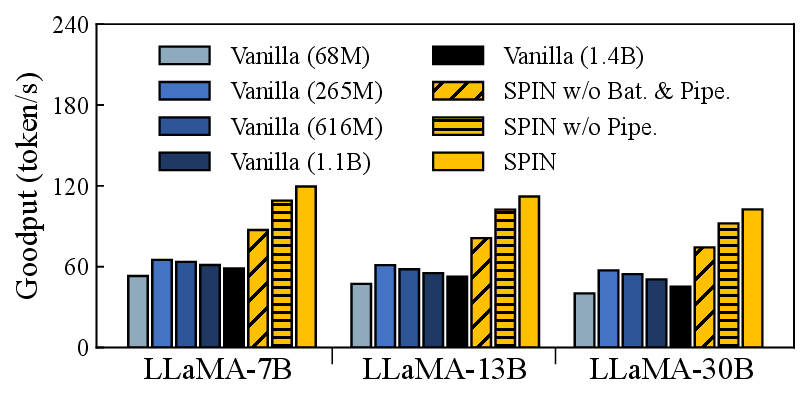} 
                \end{minipage}
    		\label{fig:ssm_cp}
    	}
            \subfigure[CIP]{
    		\begin{minipage}[b]{0.22\textwidth}
    		\includegraphics[width=1\textwidth]{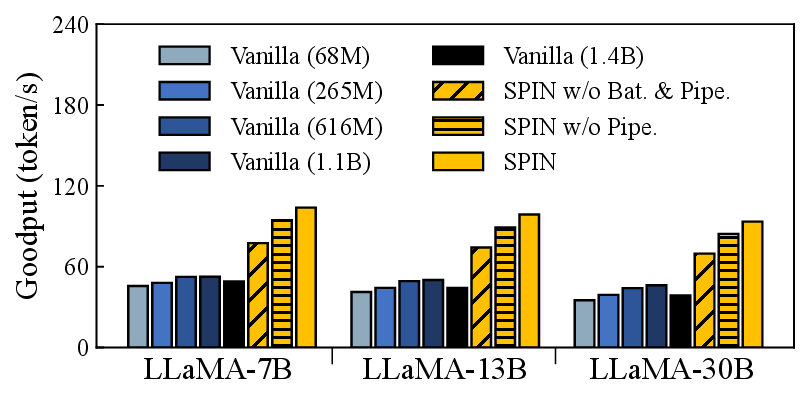} 
                \end{minipage}
    		\label{fig:ssm_cip}
    	}
            \subfigure[Mix]{
    		\begin{minipage}[b]{0.22\textwidth}
    		\includegraphics[width=1\textwidth]{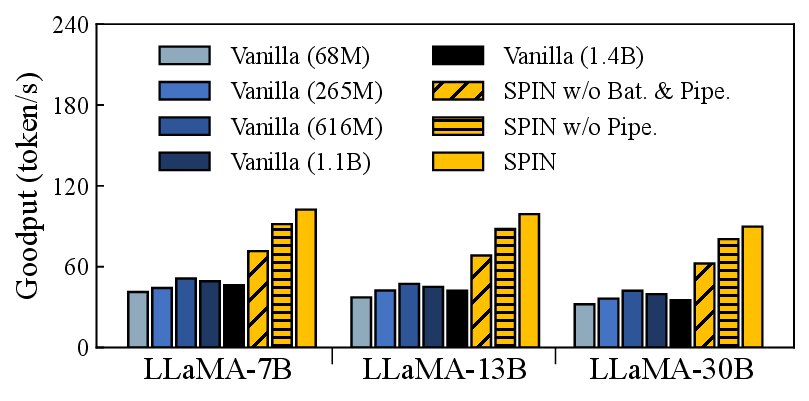} 
                \end{minipage}
    		\label{fig:ssm_mix}
    	}
    \caption{\label{fig:performance_overall}The evaluation of \textsc{Spin} on different dataset.}
    \end{figure*}
\subsection{Results}
\subsubsection{Overall Performance}
We first evaluate the overall performance of \textsc{Spin} by comparing it with the Vanilla system. We set the batch size to 8 and evaluate \textsc{Spin} in three settings: (1) \textsc{Spin} w/o bat. \& pipe.: disabling fast batch verification and pipeline optimization, only enabling heterogeneous SSM speculation; (2) \textsc{Spin} w/o pipeline: disabling pipeline optimization but adopting optimized heterogeneous speculation and fast batch verification; and (3) \textsc{Spin}: the complete implementation of \textsc{Spin}. We use goodput, the number of accepted tokens per second, as the evaluation metric. The results are shown in \autoref{fig:performance_overall}. The results show that \textsc{Spin} significantly improves goodput compared to the baselines. For example, on the Alpaca dataset with LLaMA-7B as the LLM, \textsc{Spin} improves goodput by about 2.34$\times$ compared to other baselines.

When only speculation of heterogeneous SSMs is enabled, \textsc{Spin} still achieves high performance. For example, on the Alpaca dataset with LLaMA-7B, the Vanilla speculative decoding using only LLaMA-68M achieves a goodput of 32.21 tokens per second, whereas \textsc{Spin} generates 62.33 tokens per second on average. Enabling the fast batch verification and pipeline optimizations further enhances \textsc{Spin}'s performance. For example, on the CP dataset with LLaMA-7B, enabling the fast batch verification and pipeline optimizations increases \textsc{Spin}'s performance improvement from 1.45$\times$ to 1.81$\times$ and 1.99$\times$, respectively.

In addition, we observe all baselines perform variably with different datasets. For example, the Vanilla method with LLaMA-265M works best on the CP dataset compared to others, while the Vanilla system with LLaMA-1.1B outperforms other baselines on the Alpaca dataset, as shown in \autoref{fig:ssm_alpaca} and \autoref{fig:ssm_cp}. This is because the inference requests on the Alpaca dataset exhibit higher difficulty, requiring larger SSMs to ensure the acceptance rate. In contrast, inference requests in the CP dataset can benefit more from smaller SSMs due to their lower difficulty. This observation can be also verified by the preliminary experiments in \S\ref{motivation}. However, we observe that the Vanilla speculative decoding systems with both LLaMA-68M and LLaMA-1.4B cannot outperform others and consistently show the worst performance. This can be attributed to the lowest acceptance rate with LLaMA-68M and the highest inference speed with LLaMA-1.4B.

\begin{figure*}[t]
\centering
        \subfigure[Alpaca]{
            \begin{minipage}[b]{0.22\textwidth}
            \includegraphics[width=\textwidth]{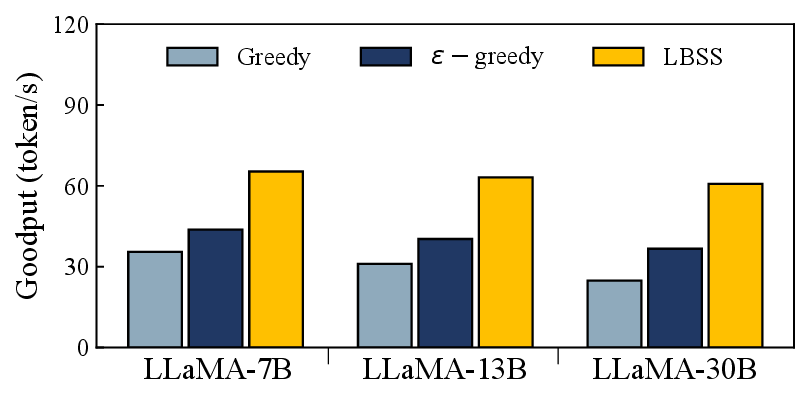}
            \end{minipage}
        \label{fig:alg_alpaca}
            }
            \subfigure[CP]{
    		\begin{minipage}[b]{0.22\textwidth}
    		\includegraphics[width=\textwidth]{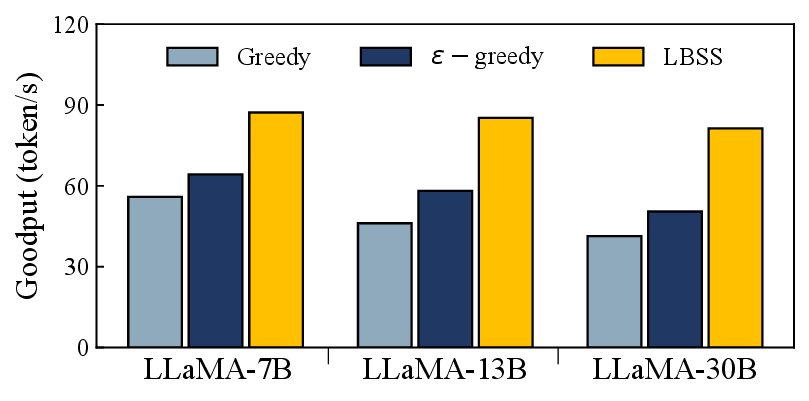} 
                \end{minipage}
    		\label{fig:alg_cp}
    	}
            \subfigure[CIP]{
    		\begin{minipage}[b]{0.22\textwidth}
    		\includegraphics[width=1\textwidth]{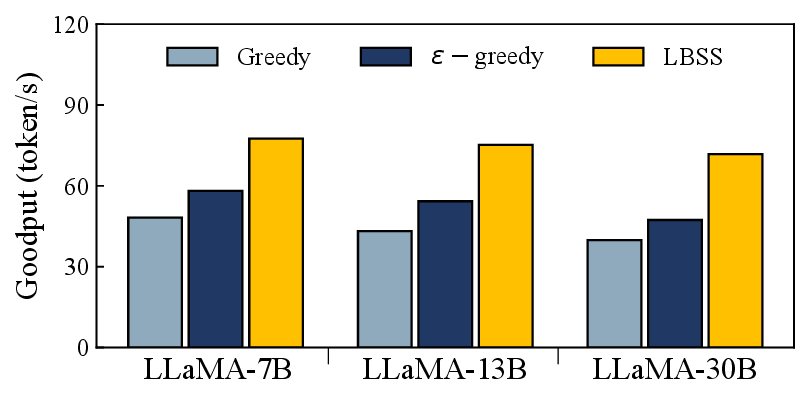} 
                \end{minipage}
    		\label{fig:alg_cip}
    	}
            \subfigure[Mix]{
    		\begin{minipage}[b]{0.22\textwidth}
    		\includegraphics[width=1\textwidth]{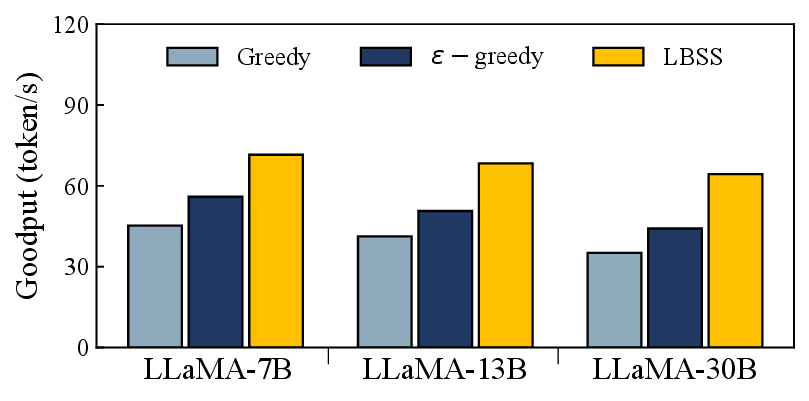} 
                \end{minipage}
    		\label{fig:alg_mix}
    	}
    \caption{\label{fig:performance_alg}The evaluation of \textsc{Spin} with different SSM selection algorithms.}
    \end{figure*}
\subsubsection{Impact of the SSM Selection Algorithm}
We compare the learning-based SSM selection (LBSS) algorithm used in \textsc{Spin} with two baselines: (1) Greedy: requests greedily select the SSMs according to their prompt lengths, where requests with shorter prompts are prioritized using smaller SSMs for speculation while others with longer prompts use the larger SSMs; and (2) $\epsilon$-greedy: with probability $\epsilon$, all inference requests use the best SSM according to the currently observed performance, and with probability $1-\epsilon$, they randomly select an SSM. We set $\epsilon=0.2$ for better exploration. Additionally, we set the batch size to 8 and disable the fast batch verification and pipeline optimization for all methods. The results are shown in \autoref{fig:performance_alg}. We observe that the proposed SSM selection algorithm in \textsc{Spin} significantly outperforms the others. For example, on the Alpaca dataset with LLaMA-7B as the LLM, the greedy method generates only 32.21 tokens per second. Although the $\epsilon$-greedy method can explore better SSM selection, it does not balance exploration and exploitation well, achieving a goodput of 43.78 tokens per second. In contrast, our LBSS effectively explores SSMs and selects the best one for inference requests, achieving a higher goodput of 65.31 tokens per second.

In addition, we observe that SSM selection becomes more crucial with larger LLMs. For example, on the CP dataset, LBSS achieves a goodput improvement of about 1.45$\times$ with the LLM of LLaMA-7B compared to the other two methods, while the improvement increases to 1.79$\times$ when using LLaMA-30B as the LLM. This can be attributed to the fact that both greedy and $\epsilon$-greedy methods tend to choose unsuitable SSMs for inference requests since they cannot estimate the goodput of SSMs accurately, resulting in significant performance degradation.

\begin{figure*}[t]
\centering
        \subfigure[Alpaca]{
            \begin{minipage}[b]{0.22\textwidth}
            \includegraphics[width=\textwidth]{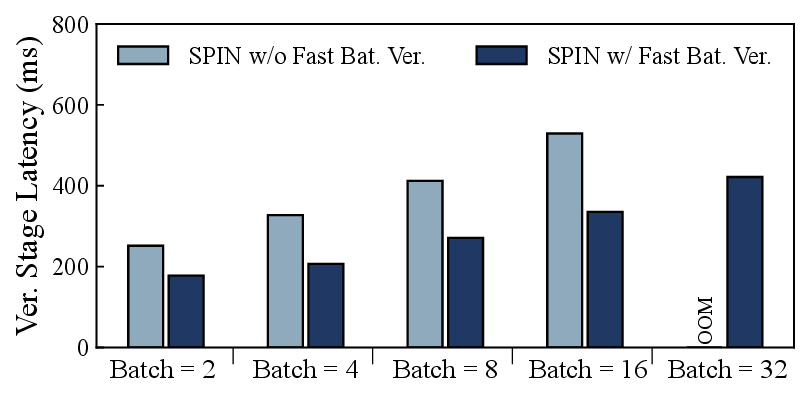}
            \end{minipage}
        \label{fig:batching_alpaca}
            }
            \subfigure[CP]{
    		\begin{minipage}[b]{0.22\textwidth}
    		\includegraphics[width=\textwidth]{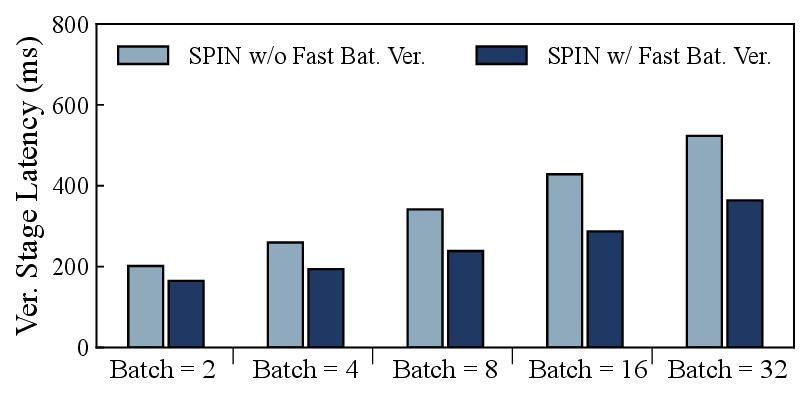} 
                \end{minipage}
    		\label{fig:batching_cp}
    	}
            \subfigure[CIP]{
    		\begin{minipage}[b]{0.22\textwidth}
    		\includegraphics[width=1\textwidth]{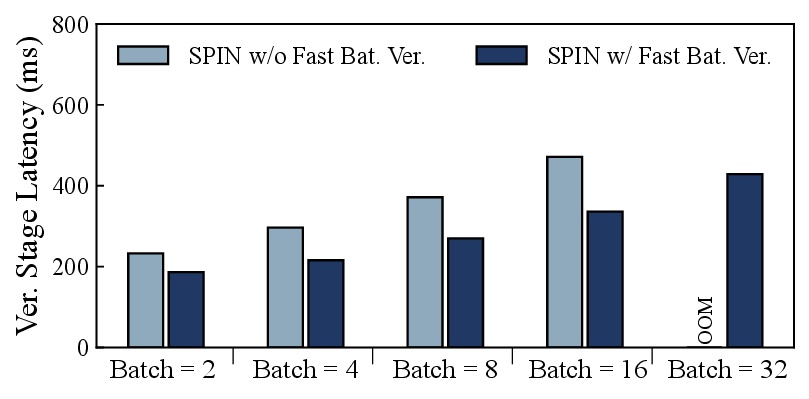} 
                \end{minipage}
    		\label{fig:batching_cip}
    	}
            \subfigure[Mix]{
    		\begin{minipage}[b]{0.22\textwidth}
    		\includegraphics[width=1\textwidth]{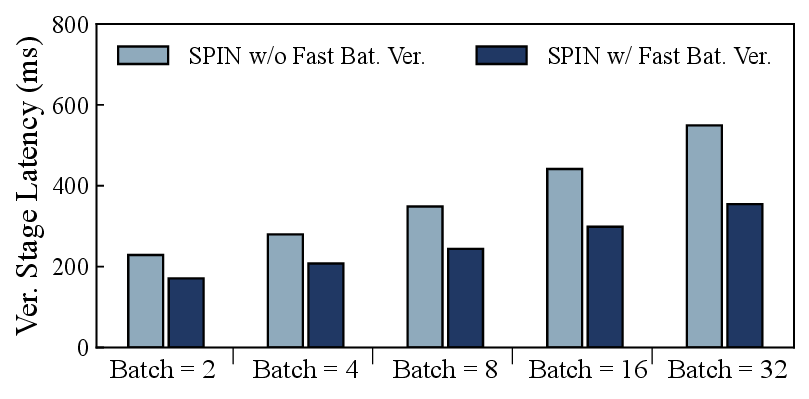} 
                \end{minipage}
    		\label{fig:batching_mix}
    	}
    \caption{\label{fig:performance_batching}The evaluation of the batch verification on different batch sizes.}
    \end{figure*}
\subsubsection{Impact of the Fast Batch Verification}
We study the impact of the proposed fast batch verification with different batch sizes, which adopts a request decomposing method to accelerate the verification on the LLM. Since the fast batch verification mainly works during the verification stage, we use the average verification latency as the evaluation metric. We use LLaMA-7B as the LLM. We change the batch size and show results in \autoref{fig:performance_batching}. We observe that the fast batch verification in \textsc{Spin} can significantly reduce verification latency and save memory. For example, on the Alpaca dataset with a batch size of 16, the fast batch verification in \textsc{Spin} reduces verification latency from 529.21ms to 335.46ms, achieving a reduction of 33.61\%.
% Overall, the batching optimization can reduce the verification latency by about 42\% with different settings. Furthermore, we observe that the batching optimization can bring higher performance improvements with larger batch sizes, which is primarily attributed to the significant number of padding tokens when the batching optimization is disabled.
When the batch size reaches 32, out-of-memory errors happen for the Alpaca and CIP datasets when the fast batch verification is disabled, because of too much memory occupation by padded tokens in batching. On the other hand, \textsc{Spin} can still work with low verification latency, thanks to its sophisticated request decomposition for memory saving. 
% of requests in these two datasets. 

In addition, we observe the fast batch verification provides the most significant performance improvement on the Alpaca dataset. This is because the inference requests exhibit high difficulty, and the number of accepted tokens can vary significantly, leading to substantial padding tokens within the batch. In contrast, the benefits of the fast batch verification are smaller on the CP dataset because the requests have more similar acceptance rates, resulting in fewer padding tokens.

\begin{figure*}[t]
\centering
        \subfigure[Alpaca]{
            \begin{minipage}[b]{0.22\textwidth}
            \includegraphics[width=\textwidth]{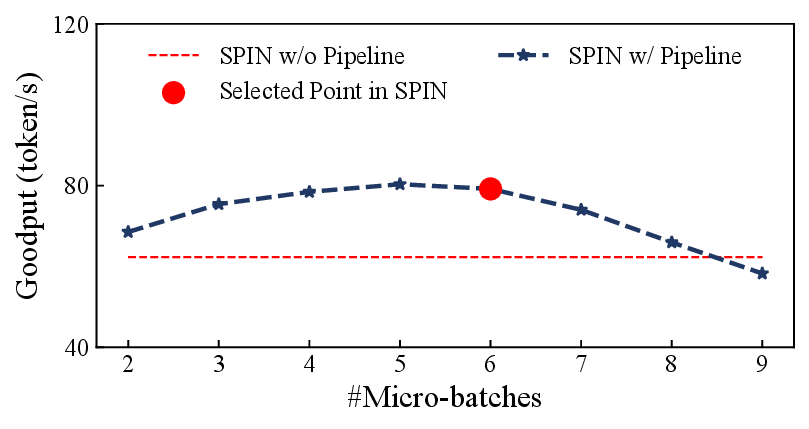}
            \end{minipage}
        \label{fig:pipeline_alpaca}
            }
            \subfigure[CP]{
    		\begin{minipage}[b]{0.22\textwidth}
    		\includegraphics[width=\textwidth]{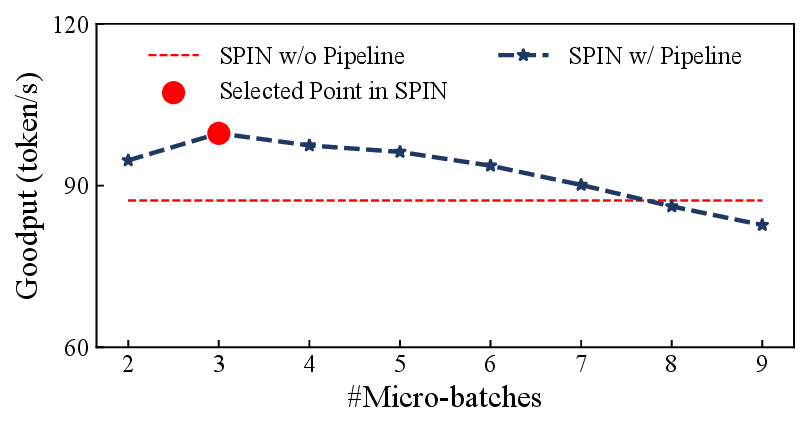} 
                \end{minipage}
    		\label{fig:pipeline_cp}
    	}
            \subfigure[CIP]{
    		\begin{minipage}[b]{0.22\textwidth}
    		\includegraphics[width=1\textwidth]{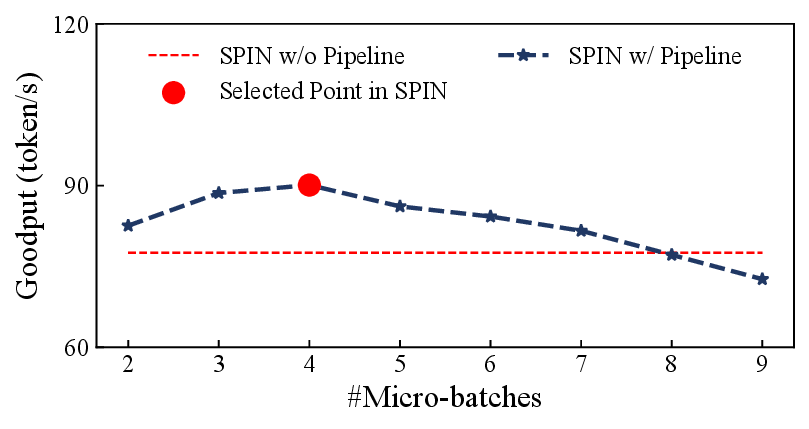} 
                \end{minipage}
    		\label{fig:pipeline_cip}
    	}
            \subfigure[Mix]{
    		\begin{minipage}[b]{0.22\textwidth}
    		\includegraphics[width=1\textwidth]{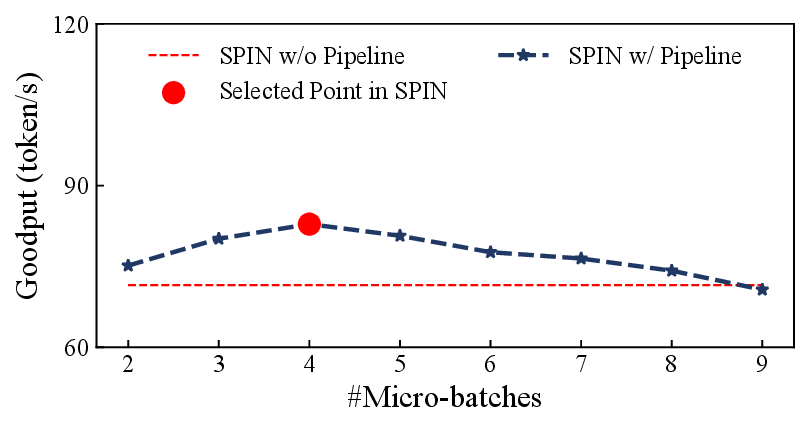} 
                \end{minipage}
    		\label{fig:pipeline_mix}
    	}
    \caption{\label{fig:performance_pipeline}The performance of speculative decoding pipeline in different numbers of micro-batches. The red dot indicates the setting used by \textsc{Spin}.}
    \end{figure*}

\subsubsection{Impact of the Speculative Decoding Pipeline}
Finally, we study the impact of the proposed speculative decoding pipeline mechanism, which divides the batch on each SSM into smaller micro-batches and enables pipeline execution of speculation and verification. We use the LLaMA-7B as the LLM and set different numbers of micro-batches on each SSM. The results are shown in \autoref{fig:performance_pipeline}. We observe that the benefits of the speculative decoding pipeline vary with different numbers of micro-batches. Specifically, the benefits increase with more micro-batches up to a certain point, and then degrade when the number of micro-batches continues to increase. For example, as shown in \autoref{fig:pipeline_alpaca}, the pipeline optimization with 4 micro-batches on each SSM can provide a performance improvement of 1.29$\times$ on the Alpaca dataset. However, the pipeline optimization does not work and even degrades speculative decoding performance when the number of micro-batches reaches 9. The rationale is as follows: With more micro-batches, the size of each micro-batch becomes smaller, reducing the variance in inference speed across SSMs. Thus, the LLM has a shorter idle time to synchronize the speculation results from SSMs to start the verification. However, when the number of micro-batches continues to increase, the speculation and verification of more requests are performed sequentially, leading to lower inference throughput.

We also observe that the optimal number of micro-batches varies across different datasets. For example, as shown in \autoref{fig:pipeline_alpaca}, the speculative decoding pipeline with 5 micro-batches works best on the Alpaca dataset, while the optimal number of micro-batches for the speculative decoding pipeline is 3 on the CP dataset, as shown in \autoref{fig:pipeline_cp}. Specifically, more inference requests in the Alpaca dataset perform the speculation on larger SSMs due to higher difficulty, such as LLaMA-1.1B and LLaMA-1.4B. Therefore, more micro-batches are needed to reduce synchronization time across SSMs. In contrast, for the CP dataset, more requests are processed by smaller SSMs, such as LLaMA-68M and LLaMA-265M, leading to smaller synchronization times across SSMs. Consequently, a smaller number of micro-batches can be sufficient to achieve optimal speculative decoding pipeline performance, and more micro-batches result in degraded inference throughput due to the sequential execution of micro-batches. Finally, we observe that the heuristic used in \textsc{Spin} to determine the number of micro-batches can approximate the optimal point well.

\section{Related Work}
\noindent\textbf{Speculative Decoding.} 
% Speculative decoding~\cite{leviathan2023fast, chen2023accelerating} accelerates LLM inference by speculating multiple tokens with a small speculative model (SSM) and then verifying them in parallel with the LLM. 
Recent works optimize the speculation phase on the SSM to enhance performance. For example, Fu et al.~\cite{fu2023lookahead} use an n-gram model to generate speculation tokens, while Eagle~\cite{li2024eagle} directly uses the embeddings from the LLM. Similarly, Medusa~\cite{cai2024medusa} trains multiple speculation heads using the output from the LLM as input. SpecInfer~\cite{miao2024specinfer} introduces the speculation tree to improve the acceptance rate of the SSM. For the verification phase, existing works mainly focus on reducing the LLM inference costs. For example, LayerSkip~\cite{elhoushi2024layer} and EESD~\cite{liu2024speculative} design the SSM as the partial of the LLM, enabling the intermediate results of the SSM to be reused in the verification phase. S3D~\cite{zhong2024s3d} adopts a mid-layer skipping method to reduce the verification cost. 
Recently, existing works propose to balance the speculation correctness and verification costs by adjusting the speculation length. For instance, Su et al.~\cite{su2023synergy} consider the relationship between the speculative window and batch size to maximize throughput, while BiLD~\cite{kim2024speculative} and Kangaroo~\cite{liu2024kangaroo} use a heuristic that stops speculation if the confidence in the speculative token falls below a predefined threshold. 
% Llmcad~\cite{xu2023llmcad} calculates the cumulative product of the confidences for speculation tokens to determine the speculation length.
However, existing works merely focus on homogeneous SSMs, limiting their ability to serve inference requests with varying difficulty.

\noindent\textbf{Speculative Decoding based Serving Systems.} Recent works aim to design efficient LLM inference serving systems based on speculative decoding. SpecInfer~\cite{miao2024specinfer} employs a tree-based parallel decoding mechanism to reduce end-to-end latency. Minions~\cite{wang2024minions} introduce an execution pipeline mechanism to decouple SSM speculation and LLM verification using multiple inference batches, thus improving resource utilization. 
% BASS~\cite{qian2024bass} introduces a parallel speculative decoder and a customized CUDA kernel to efficiently process multiple requests simultaneously. 
SmartSpec~\cite{liu2024optimizing} is a scheduling framework that determines the optimal speculation length for different requests. 
%by exploring trade-offs between acceptance rate and speculation cost under varying system loads. 
SpecExec~\cite{svirschevski2024specexec} deploys speculative decoding services on consumer hardware with RAM offloading. Although these works are promising for improving speculative decoding efficiency, they focus on serving a single inference request or a single batch using one SSM. Furthermore, they do not effectively address the challenge of the strict execution sequence between speculation and verification, which significantly degrades GPU utilization.

\section{Conclusion}
We design \textsc{Spin}, an efficient LLM inference serving system based on speculative decoding. \textsc{Spin} incorporates three novel designs. First, \textsc{Spin} exploits heterogeneous speculative models to serve inference requests with varying difficulty. Second, \textsc{Spin} accelerates batch verification by using the request decomposition method. Finally, \textsc{Spin} presents a speculative decoding pipeline mechanism by dividing the batch on each SSM into multiple micro-batches, further improving inference throughput. Extensive experiments demonstrate that \textsc{Spin} can improve inference throughput by about 2.28$\times$ compared to baseline methods.
%JD7...

\section{Acknowledgments}
This research was supported by Japan Society for the Promotion of Science Fellows No. 23KJ1786, Japan Science and Technology Agency (JST) PRESTO (No. 23828673), National Natural Science Foundation of China (No. 62471383, U23A20276 and U22A2029). Peng Li is the corresponding author.

\bibliographystyle{IEEEtran}
\bibliography{IEEEabrv,ref}

\end{document}